\documentclass[aps,prl,reprint,groupedaddress,nofootinbib,floatfix,preprintnumbers]{revtex4-1}

\usepackage{amssymb,amsmath,amsthm,mathtools}
\usepackage{xcolor}
\usepackage{hyperref}
\usepackage{slashed}
\usepackage{bm}

\usepackage{graphicx}
\usepackage[labelformat=simple]{subcaption}

\usepackage[utf8]{inputenc}

\theoremstyle{plain}

\newtheorem*{thm}{Theorem}

\theoremstyle{definition}

\newtheorem*{defn}{Definition}

\newtheoremstyle{named}{}{}{\itshape}{}{\bfseries}{.}{.5em}{#3}
\theoremstyle{named}

\begin{document}

\title{Geodesic Gradient Flows in Moduli Space}
\preprint{AFCI-T23-09}

\author{Muldrow Etheredge}
\email{muldrowdoskeyetheredge@gmail.com}
\affiliation{Department of Physics, University of Massachusetts, Amherst, MA 01003 USA}

\author{Ben Heidenreich}
\email{bheidenreich@umass.edu }
\affiliation{Department of Physics, University of Massachusetts, Amherst, MA 01003 USA}

\date{\today}

\begin{abstract}
Geodesics in moduli spaces of string vacua are important objects in string phenomenology. In this paper, we highlight a simple condition that connects brane tensions, including particle masses, with geodesics in moduli spaces. Namely, when a brane's scalar charge-to-tension ratio vector $-\nabla \log T$ has a fixed length, then the gradient flow induced by the logarithm of the brane's tension is a geodesic. We show that this condition is satisfied in many examples in the string landscape.
\end{abstract}

\pacs{}

\maketitle

\section{Introduction}
Much effort has been devoted to relating the brane and particle spectrum of a theory with its moduli space, and in particular with the geometrical properties thereof, such as the moduli space metric, closed cycles in the moduli space, and infinitely long spikes (infinite distance limits). In particular, geodesics in moduli space play a key role in several swampland conjectures, such as the (sharpened) Distance Conjecture \cite{Ooguri:2006in, Etheredge:2022opl} and the Emergent String Conjecture \cite{Lee:2019wij}, both of which constrain the behavior of the theory near its infinite distance limits. There is an extensive existing literature on these topics (e.g., \cite{Grimm:2018ohb, Blumenhagen:2018nts, Grimm:2018cpv, Corvilain:2018lgw, Joshi:2019nzi, Marchesano:2019ifh, Font:2019cxq, Erkinger:2019umg, Buratti:2018xjt, Heidenreich:2018kpg, Gendler:2020dfp, Lanza:2020qmt, Klaewer:2020lfg, vanBeest:2021lhn, Palti:2019pca, Blumenhagen:2017cxt, Baume:2016psm, Klaewer:2016kiy, Rudelius:2023mjy, Lee:2018urn, Lee:2019xtm, Lanza:2021udy, Baume:2019sry, Calderon-Infante:2023ler, Castellano:2023stg, Castellano:2023jjt}).

Local quantities, such as scalar charge-to-tension ratios, also play an important role:
\begin{defn}
	The \textbf{scalar charge-to-tension ratio} or ``$\alpha$-vector" of a brane with moduli-dependent tension $T(\phi)$ is defined as
	\begin{align}
		\vec \alpha \equiv - \vec\nabla \log T,
	\end{align}
	where the gradient is with respect to the moduli and the Planck mass is set to one.
\end{defn}
\noindent These appear prominently in the Scalar Weak Gravity Conjecture \cite{Palti:2017elp,Calderon-Infante:2020dhm,Etheredge:2023odp,Etheredge:2023usk,Etheredge:Taxonomy}, which is a convex hull condition on $\alpha$-vectors for particle towers.

Geodesics can connect asymptotic statements, such as the (sharpened) Distance Conjecture, with local ones, such as the Scalar Weak Gravity Conjecture. This was first explored in~\cite{Calderon-Infante:2020dhm}. More recently, \cite{Etheredge:2023usk} considered a mechanism where $\alpha$-vectors align with geodesics that asymptote to infinite distance limits of the moduli space.

 In this paper, we discover a very simple connection between $\alpha$-vectors and geodesics: When an $\alpha$-vector has constant length throughout moduli space, then the integral curves of the $\alpha$-vector field are geodesics. Equivalently, the gradient flows induced by the logarithm of the corresponding brane tension are geodesics.
 
In addition to proving this relation, we demonstrate that there are numerous states in string theory whose $\alpha$-vectors have fixed lengths. When the states in question are fundamental branes or infinite tower of particles, the corresponding geodesics will asymptote to infinite-distance limits in the moduli space. Our results provide a broadly applicable method for (1) finding new geodesics in moduli space, (2) connecting brane tensions with the moduli space geometry, and (3) identifying locally which geodesics go to infinite distance limits.

\section{Geodesic Gradient Flows}
We begin with a simple observation:

\begin{thm}
Consider a scalar function $F$ on the moduli whose gradient has constant length.
Then a solution to the gradient flow equation, 
\begin{align}
	\frac{d\phi^i}{dt}=\nabla^i F, \label{eqn:gradflow}
\end{align}
is a geodesic.
\end{thm}

\begin{proof}
Applying \eqref{eqn:gradflow},
we obtain 
\begin{align}
	\frac{d^2 \phi^i}{dt^2}+\Gamma^i_{jk}\frac{d\phi^j}{dt}\frac{d\phi^k}{dt} &= \partial_j(\nabla^i F)\frac{d\phi^j}{dt}+\Gamma^i_{jk}\frac{d\phi^j}{dt}\frac{d\phi^k}{dt}\nonumber\\
	&= \partial_j(\nabla^i F)\nabla^j F+\Gamma^i_{jk}\nabla^j F \nabla^k F \nonumber\\
	&= \nabla_j(\nabla^i F)\nabla^j F = \frac{1}{2}\nabla^i(\nabla F)^2 = 0 \,,
\end{align}
since $(\nabla F)^2$ is constant. Thus, $F$ gradient flows follow geodesics.
\end{proof}
Formally, this transforms the problem of solving the second-order geodesic equation into a simpler, first-order problem, at the expense of first requiring a solution to the partial differential equation
\begin{equation}
(\nabla F)^2 = \text{constant}\,. \label{eqn:Fcond}
\end{equation}

Does a function $F$ satisfying this condition necessarily exist? First, note that if we rescale $F$ such that $(\nabla F)^2 = 1$ then the affine parameter $t$ becomes the distance travelled along the flow. Moreover,
\begin{equation}
\frac{d F}{d t} = \frac{d\phi^i}{d t} \nabla_i F = (\nabla F)^2 = 1 \,,
\end{equation}
so the change in $F$ along the flow also measures the distance travelled.

Turning this around, the \emph{distance function} $d(\phi, \phi_0)$ specifying the geodesic distance between $\phi$ and fixed reference point $\phi_0$ clearly has the property that the change in $d(\phi, \phi_0)$ along a geodesic originating at $\phi_0$ equals the distance travelled. Indeed, one can show that $(\nabla d(\phi,\phi_0))^2 = 1$ (away from $\phi = \phi_0$),\footnote{To be precise, this is true at generic points; $d(\phi,\phi_0)$ (as well as the generalized distance functions discussed below) can fail to be differentiable at special points where the shortest route to $\phi_0$ changes discontinuously, e.g., when winding around a cylinder.} so \eqref{eqn:Fcond} can be solved by distance functions.

More generally, the reference point $\phi_0$ can be replaced by any region $\mathcal{R}$---such that $d_{\mathcal{R}}(\phi)$ is the length of the shortest geodesic connecting $\phi$ to $\mathcal{R}$---without affecting the condition $(\nabla d_{\mathcal{R}}(\phi))^2 = 1$ outside of $\mathcal{R}$. Choosing $\mathcal{R}$ to be an oriented codimension-one surface $\Sigma$ and $d_{\Sigma}(\phi)$ to be the \emph{signed} geodesic distance to $\Sigma$---which is positive or negative depending on which side of $\Sigma$ the geodesic approaches from---$(\nabla d_{\Sigma}(\phi))^2 = 1$ holds even atop $\Sigma$. With this generalized notion of a distance function, \emph{any} solution to $(\nabla F)^2 = 1$ is a distance function, where the reference surface is the surface $F=0$ (or, up to a shift, any other surface of constant $F$).

\subsection{Connection to branes, infinite-distance limits}

While this link between gradient flows and geodesics is intriguing, it only becomes very useful when there is a natural class of distance functions to consider.
The case of interest in this paper is when $F$ is the (negative) logarithm of the tension $T$ of a brane
\begin{align}
	F=-\log T .
\end{align}
In this case, $\vec{\nabla} F$ is the brane's $\alpha$-vector, $\vec{\alpha} = - \vec{\nabla} \log T= \vec{\nabla} F$. Thus, when $\vec{\alpha}$ has fixed length, $-\log T$ is a distance function and gradient flows induced by $-\log T$ are geodesics




Let us assume that the species scale can only go to zero in infinite-distance limits of the moduli space. Then, whenever a tower of particles or a fundamental string (or a ``fundamental'' $p$-brane, see the discussion in~\cite{Reece:2018zvv}) has an $\alpha$-vector of constant length everywhere in moduli space, the resulting geodesic gradient flow is a geodesic going to an infinite-distance limit. That is, the geodesic gradient flow cannot spiral around the moduli space endlessly but must eventually asymptote to infinite distance. 

This follows because the $\alpha$-vectors from these states align with the geodesic gradient flows they generate. This in turn implies that, along the geodesic $\phi^i(t)$ generated by the gradient flow,
\begin{align}
	\frac{d}{dt} (-\log T) = -\frac {d\phi^i}{dt}\nabla _i\log T=\alpha^2.
\end{align}
Starting at $\phi_0$ and integrating a distance $d(\phi,\phi_0)$ along the geodesic to the point $\phi$, the mass of the tower / tension of the brane decreases exponentially:
\begin{align}
	m(\phi)= e^{-|\alpha|d(\phi,\phi_0)}m(\phi_0),\quad T(\phi)=e^{-|\alpha|d(\phi,\phi_0)}T(\phi_0).
\end{align}
Thus, the tower of particles and/or the brane oscillation modes become parametrically light along the gradient flow geodesic, sending the species scale $\Lambda$ to zero. By assumption, this implies that we are approaching an infinite-distance limit in the moduli space.

Because of this, the particle and string spectrum can be used to determine locally a sufficient condition for a geodesic to be an infinite-distance geodesic, provided that the relevant $\alpha$-vectors have fixed length. 

\section{Examples}

In fact, particles and branes with fixed-length $\alpha$-vectors are common in string/M-theory, especially in theories with 16 or 32 supercharges. We now showcase several of these examples.

\subsection{10d examples}

In 10d theories, the 1/2 BPS fundamental strings, 1/2 BPS D$p$-branes, and 1/2 BPS NS5-branes have tensions that depend on the canonically normalized dilaton $\phi$ via
\begin{align}
	T_{\text{D}p}\sim e^{\frac{p-3}{\sqrt{d-2}}\phi},\quad T_\text{F1}\sim e^{\frac{2}{\sqrt{d-2}}\phi},\quad
	T_\text{NS5}\sim e^{-\frac{2}{\sqrt{d-2}}\phi},
\end{align}
in 10d Planck units. The resulting $\alpha$-vectors for fundamental strings, D-branes, and NS5-branes are all of constant length, and thus the gradient flows of the logarithms of these tensions are geodesics.

Additionally, oscillators of $p$-branes scale as
\begin{align}
	m_\text{osc}\sim T_p^{\frac 1{p+1}} \qquad\Longrightarrow\qquad \vec \alpha_\text{osc}=\frac 1{p+1}\vec\alpha_p \,,
\end{align}
hence these $\alpha$-vectors also have fixed length, provided that the branes have $\alpha$-vectors of fixed length. When this happens, the gradient flows of logarithms of masses of oscillators are the same geodesics generated by the gradient flows of the logarithms of the tensions.



In 10d IIB string theory, constant-length $\alpha$-vectors generate infinitely many independent geodesics \cite{Etheredge:2023usk}. The moduli space here is parametrized by the axiodilaton
\begin{align}
	\tau=C_0+ie^{-\phi},
\end{align}
valued in the fundamental domain
\begin{align}
	\mathcal M_\text{10d IIB}=\left\{\tau=\tau_1+i\tau_2\ |\ \tau_1\in[-1/2,1/2], |\tau|\geq 1\right\},
\end{align}
where the metric is the Poincar\'e upper half-plane metric,
\begin{align}
	ds^2=\frac{d\tau_1^2+d\tau_2^2}{\tau_2^2}.
\end{align}

The tensions of $(p,q)$ strings and fivebranes can be obtained from the SL$(2,\mathbb Z)$-orbits of the D1 and D5-branes. The $(p,q)$ strings and fivebranes have tensions of:
\begin{align}
	T_{(p,q)}^{1,5}&\sim \sqrt{e^\Phi (p+C_0q)^2+e^{-\Phi}q^2}=\frac{|p+\tau q|}{\sqrt{\tau_2}}.
\end{align}

Using a canonically normalized basis,
the scalar charge-to-tension ratios of $(p,q)$ strings and fivebranes are \cite{Etheredge:2022opl,Etheredge:2023usk}
\begin{align}
	\vec \alpha_{(p,q)}^\text{strings}=\vec \alpha_{(p,q)}^\text{5-branes}=\begin{pmatrix}\frac{\sqrt2q\tau_2(p+\tau_1q)}{|p+\tau q|^2}\\\frac{q^2\tau_2^2-(p+\tau_1q)^2}{\sqrt 2|p+\tau q|^2}\end{pmatrix}.\label{e.10dIIBpqalphas}
\end{align}
The $(p,q)$ strings and fivebranes all have constant length $2/\sqrt{d-2}=1/\sqrt{2}$.

The logarithms of tensions of $(p,q)$ strings and fivebranes generate infinitely many independent infinite-distance geodesic gradient flows, since $p$ and $q$ can be varied over infinitely many values. See Figures \ref{f.10dIIBpqfield} and \ref{f.10dIIBpqalpharaygeodesics}. 

\begin{figure}
\begin{center}
\includegraphics[width = 80mm]{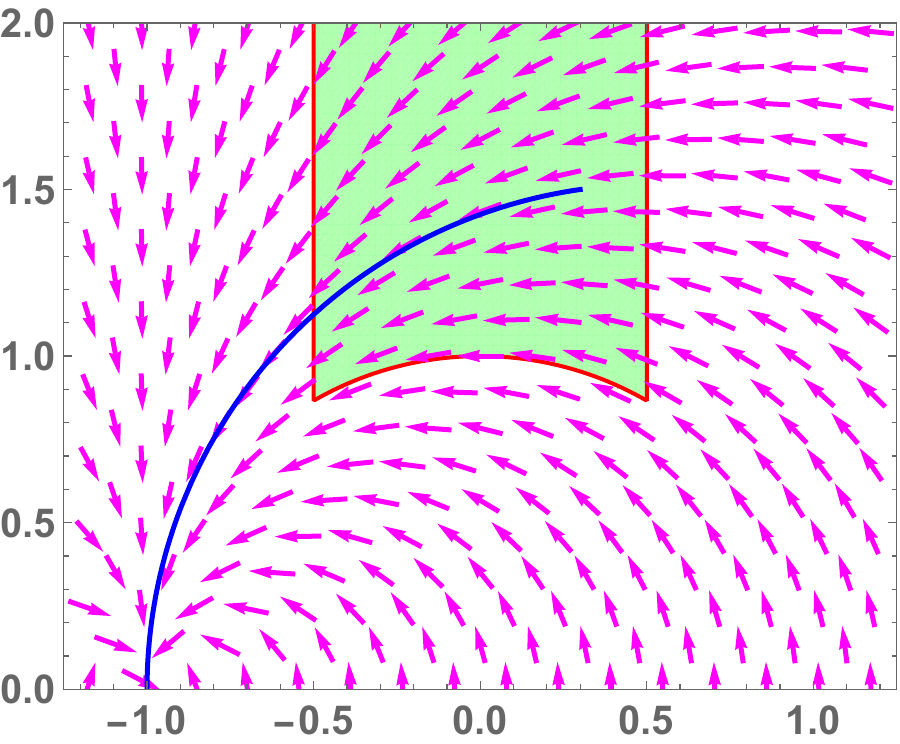}
\end{center}
\caption{$\alpha$-vector field for $(1,1)$ strings on covering space. The blue line is an infinite-distance geodesic going to the point $-1$ on the $\tau_1$-axis of the Poincar\'e half-plane covering space. Figure from \cite{Etheredge:2023usk}.}
\label{f.10dIIBpqfield}
\end{figure}

\begin{figure}
\begin{center}
\includegraphics[width = 80mm]{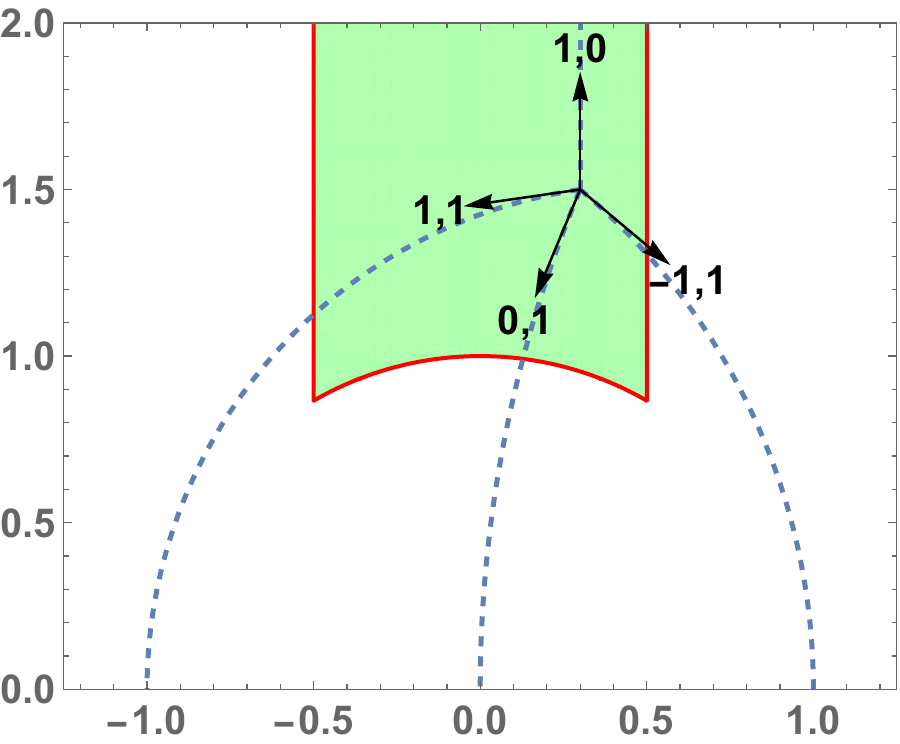}
\end{center}
\caption{$\alpha$-vectors for $(1,1)$, $(-1,1)$, $(0,1)$ and $(1,0)$ strings/fivebranes are parallel to geodesics going to $\tau_1=-1$, $1$, $0$, and $i\infty$ on $\tau_2=0$-axis of Poincar\'e covering space. Figure from \cite{Etheredge:2023usk}. }
\label{f.10dIIBpqalpharaygeodesics}
\end{figure}

\subsection{9d examples}

For heterotic string theory on a circle, both the wrapped and unwrapped fundamental strings, as well as the KK-modes, have constant-length $\alpha$-vectors \cite{Etheredge:2023odp}. So too do the wrapped and unwrapped NS5-branes, and also the KK-monopole. The dilaton-radion components of these are depicted in Figure \ref{f.9dhetalphas}.

\begin{figure}
\begin{center}
\includegraphics[width = 80mm]{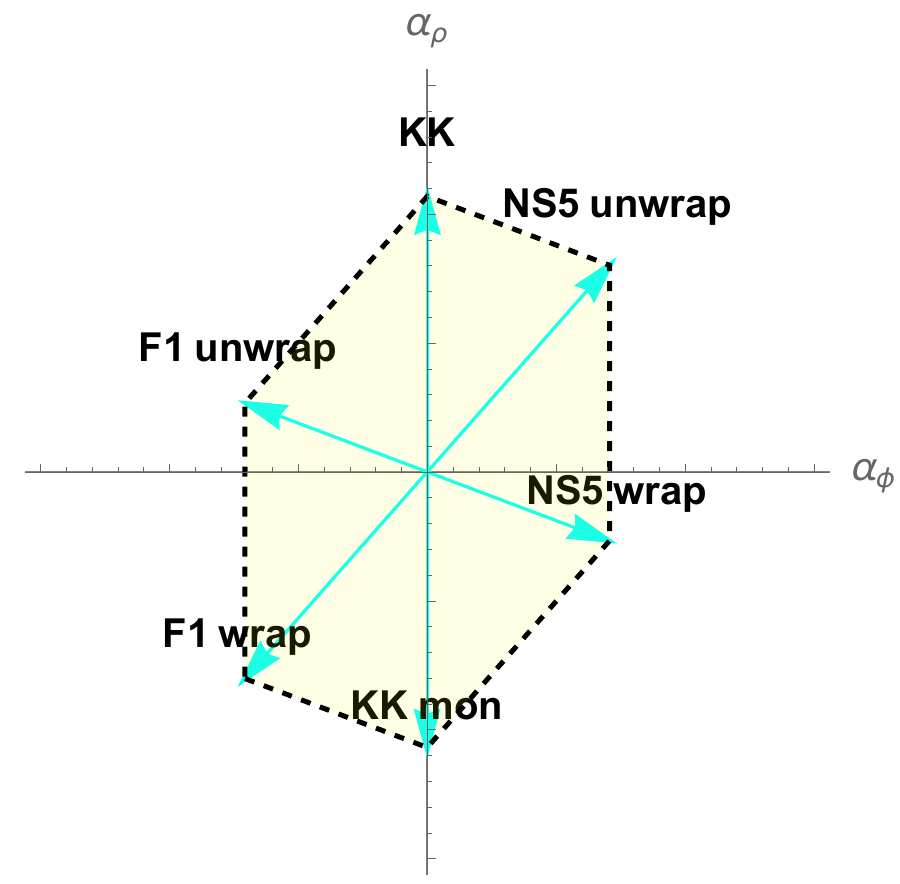}
\end{center}
\caption{Dilaton-radion components of some constant-length $\alpha$-vectors from 9d heterotic string theory. Depicted are KK modes, the KK monopole, and wrapped and unwrapped fundamental strings and NS5-branes.}
\label{f.9dhetalphas}
\end{figure}

Next, consider M-theory on a 2-torus, where the moduli space is three-dimensional and non-flat. The moduli space is parametrized by the volume $U$ of the torus and its shape parameters $\tau=\tau_0+i\tau_1$ \cite{Etheredge:2022opl}. Alternatively, from a IIB on a circle perspective, $\tau$ is IIB string theory's axiodilaton, and $U$ is the circle's radion.

There are infinitely many states with constant $\alpha$-vector lengths. By \cite{Etheredge:2023usk}, the KK-modes with $p$-quanta of momentum along one cycle and $q$-quanta of momenta along the other cycle have constant length $\alpha$-vectors. The KK-monopoles have tensions that are given by just the reciprocals of the KK-mode masses and thus have $\alpha$-vectors of constant length. For M5-branes and M2-branes unwrapped or fully wrapped, the tensions scale with the moduli through exponentials of $U$, and thus their $\alpha$-vectors are of constant length.\footnote{Note that the 1/4 BPS consisting of wrapped M2 branes carrying KK momentum do \emph{not} have fixed length $\alpha$-vectors, see, e.g.,~\cite{Etheredge:2022opl}. There are of course many other states like this in string compactifcations; here we are focusing on the states which do have fixed length $\alpha$-vectors.\label{fn:quaterBPS}} IIB string theory's $(p,q)$ strings and fivebranes describe the strings and particles that come from M2-branes and M5-branes that wrap one cycle  $p$ times and the other by $q$ times, and thus also have constant-length $\alpha$-vectors. The radion-radion components of some of these branes for an orthogonal torus is plotted in Figure \ref{f.9dIIAalphas}.

\begin{figure}
\begin{center}
\includegraphics[width = 80mm]{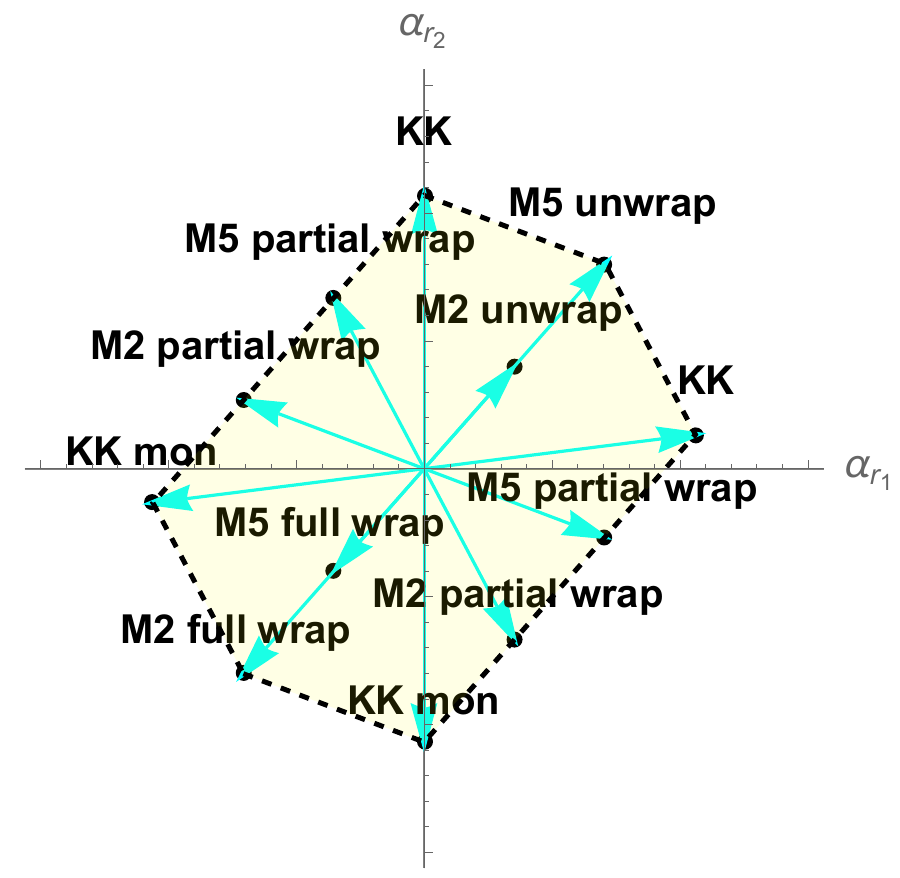}
\end{center}
\caption{Radion-radion components of constant-length $\alpha$-vectors from M-theory on $T^2$. Depicted are the KK modes and monopoles, as well as unwrapped, partially wrapped, and fully wrapped M2 and M5-branes.}
\label{f.9dIIAalphas}
\end{figure}

\subsection{8d examples and beyond}

In the case of heterotic string theory on a two-torus, states with fixed-length $\alpha$-vectors include the 1/2 BPS KK-modes, fully or partially wrapped fundamental strings, NS5-branes, and KK-monopoles. Likeiwse, for M-theory on $T^3$, states with fixed-length $\alpha$-vectors include 1/2 BPS KK modes, fully-wrapped, partially-wrapped, and fully unwrapped M2-branes, M5-branes, and KK-monopoles.

This is just a small sample. There are, of course, numerous other examples in string/M/F-theory compactifications; we will not attempt to list them all.

\section{Future directions}

Our results raise a number of interesting questions---discussed below---whose resolution we leave to future work.

\subsection{Which states have fixed-length $\alpha$-vectors?}

All of the examples of states with fixed-length $\alpha$-vectors that we have discussed are 1/2 BPS states in maximal and half-maximal supergravity theories. It would be interesting to understand whether the phenomenon extends beyond this particular arena. Note that, even in maximal supergravity, 1/4 BPS states \emph{do not} have fixed-length $\alpha$-vectors, see footnote \ref{fn:quaterBPS}.

Moreover, there are certainly examples of 1/2 BPS states in theories with only 8 supercharges which do not have $\alpha$-vectors of fixed length. For example, this is common for the 5d 1/2 BPS particles arising from M2 branes wrapping curves on a Calabi-Yau threefold, as well as for the 1/2 BPS strings arising from M5 branes wrapping divisors. Nonetheless, in special circumstances these states could still have fixed-length $\alpha$-vectors; we leave further investigation of this to future work.

A related question is when particles and/or branes have $\alpha$-vectors of approximately fixed length in some asymptotic region. We expect that this is very common, if not ubiquitous (and therefore our methods could have general applicability in such regions) but we defer this question to future work as well.

\subsection{Self-force and constant gauge charge-to-tension}

BPS states, of course, have vanishing self-force. More generally, for any state with vanishing self-force, $\vec{\alpha}$ has fixed length if and only if the (canonically normalized) gauge charge-to-tension ratio has fixed magnitude. This is because $(p-1)$-branes with zero long-range self-force have gauge charge-to-tension ratios that satisfy (in $d$-dimensions with $\kappa_d=1$) \cite{Heidenreich:2020upe}
\begin{align}
	\frac{Q^2}{T^2}=\alpha^2+\frac{p(d-p-2)}{d-2}.
\end{align}
Thus, for states with zero self-force (e.g., extremal branes at two-derivative order~\cite{Heidenreich:2019zkl,Heidenreich:2020upe,Harlow:2022ich,Etheredge:2022rfl}), $\vec{\alpha}$ has fixed length if and only if $\vec{Q}/T$ has fixed magnitude. Indeed, $\frac{Q^2}{T^2} = 2$ for all the basic fundamental branes in string/M-theory, but it would be interesting to understand what happens in more involved settings.

We briefly note a few immediate consequences of this observation. Firstly, if the charge-to-tension ratio is constant, then $-\log Q^2(\phi)$ (like $-\log T$) must be a distance function. This imposes non-trivial ($Q$-dependent) constraints on a combination of the gauge-kinetic matrix and the metric on moduli space, limiting the class of theories in which this confluence between a constant $\alpha$-vector and zero self-force can occur. 

Secondly, since $-\log T$ is a distance function, vanishing tension $T\to 0$ can only occur in at infinite distance, $-\log T \to \infty$. Assuming vanishing self-force, this implies that the state is extremal, see~\cite{Harlow:2022ich}, appendix A. The properties of the corresponding extremal black hole / black brane solution are controlled by the fact that $-\log Q^2(\phi)$ is also a distance function. Since $Q^2(\phi)$ has no local minima, the solution cannot flow to an attractor point inside the moduli space. Instead, the attractor flows go to infinite distance, where necessarily $Q^2(\phi) \to 0$ since the distance function $-\log Q^2(\phi)$ diverges. Thus, the horizon area vanishes and the moduli diverge at the horizon, much like the basic extremal black brane solutions in string theory.

\subsection{Laplacians in moduli space}
In addition to the gradients of the logarithms of tensions and masses, it is interesting to consider the Laplacians of such functions (equivalently, the divergences of the $\alpha$-vectors). In many of the examples we have studied above where $\alpha$-vectors have constant length, the divergences of such $\alpha$-vectors are constant. 

In 10d IIB string theory, $(p,q)$ strings and fivebranes have tensions satisfying
\begin{align}
	\nabla^2 \log T_{(p,q)}^{1,5}=\frac 12.
\end{align}
Likewise, in 10d type I, IIA, and heterotic theories, all of the 1/2 BPS brane tensions satisfy\footnote{Since the moduli space is one-dimensional, this is a trivial consequence of $\vec{\alpha}$ having fixed length.}
\begin{align}
	\nabla^2 \log T=0.
\end{align}
In 9d maximal supergravity, the particles from fully-wrapped M2-branes, and the wrapped and unwrapped $(p,q)$ strings from IIB, respectively satisfy
\begin{align}
	\nabla^2 \log m_w=0,\quad
	\nabla^2 \log m_{(p,q)}=1,\quad
	\nabla^2 \log T^{(1)}_{(p,q)}=1.
\end{align}
It would be interesting to explore this further.

Note that when both $\vec{\nabla}\cdot\vec{\alpha} = -\nabla^2 \log T$ and $|\vec{\alpha}| = |\nabla \log T|$ are constant, the tension is necessarily an eigenfunction of the Laplacian:
\begin{align}
	\nabla^2 T=[(\nabla \log T)^2+\nabla^2 \log T]T= (\alpha^2+\rho)T,
\end{align}
where $\rho=\nabla^2 \log T$. 

These results suggest the possibility of additional deep connections between the particle/brane spectrum, distance functions, and the moduli space geometry. Further exploring Laplacian eigenfunctions, eigenvalues, and solutions to Laplace's equation may lead to further insights.


\subsection{Distance functions}
As discussed in this paper, when an $\alpha$-vector has fixed length, then $\log T$ is, up to a multiplicative factor, a distance function.

However, generic distance functions can have codimension-one (or higher codimension) loci where they are non-differentiable and the equation $(\nabla F)^2 = 1$ fails. For instance, on a circle $\theta\in[0,2\pi)$, the distance function with reference point $\theta=0$ has non-differentiable kinks at $\theta=0, \pi$.\footnote{By contrast, the signed distance function has no kink, but is discontinuous at $\theta = \pi$.} These kinks occur where the shortest geodesic jumps discontinuously (or shrinks to zero length).

When these kinks occur, $F$ gradient flows only coincide with geodesics away from the kinks. By contrast, in the examples we have discussed the brane tension is everywhere analytic, so that $\log T$ is a distance function without any kinks in it and gradient flows coincide perfectly with geodesics.

It would be thus interesting to better understand what kinds of distance functions have well-defined gradients everywhere in moduli spaces, and whether the existence of such a distance function places any constraints on the moduli space geometry.

While we defer further exploration of this to future work, we note in passing that a compact moduli space cannot have a kink-free distance function defined on it because any function on a compact space has a global minimum, and $\nabla F = 0$ at such a minimum (if $F$ is differentiable), contradicting $(\nabla F)^2 = 1$. Moreover, there are non-compact (incomplete) metric spaces, such as a sphere with a single point removed, that still do not admit kink-free distance functions. Perhaps the non-trivial requirement of the existence of a kink-free distance function could be related to some of the conjectural properties of quantum gravity moduli spaces~\cite{Ooguri:2006in}, such as non-compactness and (asymptotic) negative curvature.

\subsection{$\alpha$-vectors of branes}
Our results motivate the study of $\alpha$-vectors for branes, since we have shown that these $\alpha$-vectors are often connected with moduli space geodesics. Brane $\alpha$-vectors have been comparatively little explored in the literature (see however~\cite{Font:2019cxq}). This is in contrast to the study of $\alpha$-vectors for particles and strings, which have been extensively studied, and constrained \cite{Etheredge:Taxonomy}, due to their connection with the Distance Conjecture, Emergent String Conjecture, and Scalar Weak Gravity Conjecture. We explore brane $\alpha$-vectors further in~\cite{Etheredge:BraneSWGC}.

\section{Conclusions}

In this paper, we have shown that when the scalar charge-to-tension ratio of a brane has constant length, the gradient flow induced by the logarithm of the tension is a geodesic. We have demonstrated that this frequently occurs in the string landscape. When this happens for particle towers or fundamental branes, our work implies that the resulting geodesics reach an infinite distance limit. Our results provide a method for generating geodesics, and for diagnosing which geodesics travel to infinite-distance limits, connecting local data (e.g., scalar charge-to-tension ratios) with asymptotic properties (such as the distance conjecture).

\begin{acknowledgments}
\vspace*{1cm}
{\bf Acknowledgments.}  
We are grateful for conversations with Naomi Gendler, Jacob McNamara, Matthew Reece, Tom Rudelius, Ignacio Ruiz, and Irene Valenzuela. This work was supported by NSF grant PHY-2112800.
\end{acknowledgments}

\bibliography{ref}

\end{document}